\documentclass[a4paper,UKenglish]{lipics-v2016}

\graphicspath{{./figures/}}
\usepackage{amsthm}
\usepackage{amsmath,amssymb}
\usepackage{stmaryrd}
\usepackage{scrextend}
\usepackage{thmtools}
\usepackage{thm-restate}
\usepackage{todonotes}
\usepackage{mathtools}
\addtokomafont{labelinglabel}{\sffamily}

\newcommand{\citeauthor}[1]{\cite{#1}}
\newcommand{\citep}[1]{\cite{#1}}
\newcommand{\citet}[1]{\cite{#1}}
\usepackage{booktabs}
\usepackage{microtype}
\usepackage{xspace}

\newcommand{\CClique}{\textsc{$k$-Multicolored Clique}\xspace}

\newcommand{\Mosigmaone}{\textsc{MC($\Sigma_1$)}\xspace}
\newcommand{\Sigmaone}{$\Sigma_1$\xspace}

\newcommand{\Customized}{\textsc{MC($\forall^{\neq}$-FO)}\xspace}
\newcommand{\MCFO}{\textsc{MC(\FO)}\xspace}
\newcommand{\neqFO}{\textsc{$\forall^{\neq}$-FO}\xspace}

\newcommand{\gamename}{\textsc}
\newcommand{\classname}[1]{\textsf{\textup{\small #1}}\xspace}

\newcommand{\fpt}{\classname{FPT}}
\newcommand{\np}{\classname{NP}}

\newcommand{\pspace}{\classname{PSPACE}}
\newcommand{\wone}{\classname{W[1]}}
\newcommand{\cowone}{\classname{co-W[1]}}
\newcommand{\classw}{\classname{W}}
\newcommand{\awstar}{\classname{AW[*]}}

\newcommand{\SGG}{\textsc{Short Generalized Geography}\xspace}
\newcommand{\Hex}{\textsc{Hex}\xspace}
\newcommand{\GHex}{\textsc{Generalized Hex}\xspace}
\newcommand{\SHex}{\textsc{Short Hex}\xspace}
\newcommand{\SGHex}{\textsc{Short Generalized Hex}\xspace}
\newcommand{\SMM}{\textsc{Short Maker-Maker}\xspace}
\newcommand{\SMB}{\textsc{Short Maker-Breaker}\xspace}

\newcommand{\SEA}{\textsc{Short Enforcer-Avoider}\xspace}
\newcommand{\ConnectNoP}{\textsc{Connect}}
\newcommand{\Connect}{\textsc{Connect}($m$, $n$, $k$, $p$, $q$)\xspace}
\newcommand{\SConnect}{\textsc{Short $k$-Connect}\xspace}

\newcommand{\gpath}[1]{\mathit{path}(#1)}
\newcommand{\galigned}[1]{\mathit{hor\_vert}(#1)}
\newcommand{\gdaligned}[1]{\mathit{diag}(#1)}
\newcommand{\gAligned}[1]{\mathit{aligned}(#1)}

\newcommand{\delae}[2]{D_{#1}^{\exists}(\{#2\})}
\newcommand{\delaa}[2]{D_{#1}^{\forall}(\{#2\})}
\newcommand{\delai}[3]{D_{#2}^{#1}(#3)}
\newcommand{\vdelay}[2]{V^D_{#1}(\{#2\})}
\newcommand{\delaeb}[2]{D_{#1}^{\exists}(#2)}
\newcommand{\delaab}[2]{D_{#1}^{\forall}(#2)}
\newcommand{\vdelayb}[2]{V^D_{#1}(#2)}
\newcommand{\vexists}[1]{V^{\exists}(#1)}
\newcommand{\vforall}[1]{V^{\forall}(#1)}

\newcommand{\FO}{\textsc{FO}}
\newcommand{\card}[1]{\left| #1 \right|} %

\newcommand{\paramdecisionproblem}[4]{
\begin{center}
\begin{tabular}{|r@{\;}p{0.79\columnwidth}|}
\hline
\multicolumn{2}{|l|}{#1} \\
\textit{Instance:} & #2 \\
\textit{Parameter:} & #3 \\
\textit{Problem:} & #4\\
\hline
\end{tabular}
\end{center}
}

\usepackage{comment}%

\title{The Parameterized Complexity of Positional Games}%
\titlerunning{The Parameterized Complexity of Positional Games} %

\author[1]{Édouard Bonnet}
\author[2,3]{Serge Gaspers}
\author[4]{Antonin Lambilliotte}
\author[2,5]{Stefan R\"ummele}
\author[2]{Abdallah Saffidine}
\affil[1]{Middlesex University, London, UK\\
  \texttt{edouard.bonnet@dauphine.fr}}
\affil[2]{The University of New South Wales, Sydney, Australia\\
  \texttt{sergeg@cse.unsw.edu.au}, \texttt{s.rummele@unsw.edu.au}, \texttt{abdallah.saffidine@gmail.com}}
\affil[3]{Data61, CSIRO, Sydney, Australia}
\affil[4]{École Normale Supérieure de Lyon, Lyon, France\\
  \texttt{antonin.lambilliotte@ens-lyon.fr}}
\affil[5]{The University of Sydney, Sydney, Australia}
\authorrunning{É. Bonnet, S. Gaspers, A. Lambilliotte, S. Rümmele, and A. Saffidine}

\Copyright{Édouard Bonnet, Serge Gaspers, Antonin Lambilliotte, Stefan R\"ummele, and Abdallah Saffidine}


\begin{document}

\maketitle

\begin{abstract}
  We study the parameterized complexity of several positional games.
  Our main result is that \SGHex is \wone-complete parameterized by the number of moves.
  This solves an open problem from Downey and Fellows' influential list of open problems from 1999.
  Previously, the problem was thought of as a natural candidate for \awstar-completeness.

  Our main tool is a new fragment of first-order logic where universally quantified variables only occur in inequalities.
  We show that model-checking on arbitrary relational structures for a formula in this fragment is \wone-complete when parameterized by formula size.

  We also consider a general framework where a positional game is represented as a hypergraph and two players alternately pick vertices.
  In a Maker-Maker game, the first player to have picked all the vertices of some hyperedge wins the game.
  In a Maker-Breaker game, the first player wins if she picks all the vertices of some hyperedge, and the second player wins otherwise.
  In an Enforcer-Avoider game, the first player wins if the second player picks all the vertices of some hyperedge, and the second player wins otherwise.

  \SMM is \awstar-complete, whereas \SMB is \wone-complete and \SEA \cowone-complete parameterized by the number of moves.
  This suggests a rough parameterized complexity categorization into positional games that are complete for the first level of the \classw-hierarchy when the winning configurations only depend on which vertices one player has been able to pick, but \awstar-completeness when the winning condition depends on which vertices both players have picked.
  However, some positional games where the board and the winning configurations are highly structured are fixed-parameter tractable.
  We give another example of such a game, \SConnect, which is fixed-parameter tractable when parameterized by the number of moves.
\end{abstract}

\section{Introduction}

In a \emph{positional game} \cite{HalesJ63}, two players alternately claim unoccupied elements of the board of the game.
The goal of a player is to claim a set of elements that form a winning set, and/or to prevent the other player from doing so.

\gamename{Tic-Tac-Toe}, its competitive variant played on a $15\times 15$ board, \gamename{Gomoku}, as well as \gamename{Hex} are the most well-known positional games.
When the size of the board is not fixed, the decision problem, whether the first player has a winning strategy from a given position in the game is \pspace-complete for many such games.
The first result was established for \GHex, a variant played on an arbitrary graph~\cite{EvenT1976}.
\citeauthor{Reisch1980} soon followed up with results for \gamename{gomoku}~\cite{Reisch1980} and \gamename{Hex} played on a board \cite{Reisch81}.
More recently, \pspace-completeness was obtained for \gamename{Havannah}~\cite{BonnetJS2016TCS} and several variants of \gamename{Connect}($m$, $n$, $k$, $p$, $q$)~\cite{HsiehT2007}, a framework that encompasses \gamename{Tic-Tac-Toe} and \gamename{Gomoku}.

In a Maker-Maker game, also known as strong positional game, the winner is the first player to claim all the elements of some winning set.
In a Maker-Breaker game, also known as weak positional game, the first player, Maker, wins by claiming all the elements of a winning set, and the second player, Breaker, wins by preventing Maker from doing so.
In an Enforcer-Avoider game, the first player, Enforcer, wins if the second player claims all the vertices of a winning set, and the second player, Avoider, wins otherwise.

In this paper, we consider the corresponding short games, of deciding whether the first player has a winning strategy in $\ell$ moves from a given position in the game, and parameterize them by $\ell$.
The parameterized complexity of short games is known for games such as generalized chess \cite{ScottS08}, generalized geography \cite{AbrahamsonDF93,AbrahamsonDF95}, and pursuit-evasion games \cite{ScottS10}.
For \Hex, played on a hexagonal grid, the short game is \fpt and for \GHex, played on an arbitrary graph, the short game is \wone-hard and in \awstar.

When winning sets are given as arbitrary hyperedges in a hypergraph, we refer to the three game variants as \gamename{Maker-Maker}, \gamename{Maker-Breaker}, and \gamename{Enforcer-Avoider}, respectively.
\gamename{Maker-Breaker} was first shown \pspace-complete by \citet{Schaefer1978} under the name $G_{\text{pos}}(\text{POS DNF})$.
A simpler proof was later given by \citet{Byskov2004} who also showed \pspace-completeness of \gamename{Maker-Maker}.
To the best of our knowledge, the classical complexity of \gamename{Enforcer-Avoider} has not been established yet.

In this paper we will show that the short game for \GHex is \wone-complete, solving an open problem stated numerous times \cite{BonnetJS2016TCS,DowneyF2013,DowneyF99,FominM12,Scott2009}, we establish that the short game for a generalization of Tic-Tac-Toe is \fpt, and we determine the parameterized complexity of the short games for \gamename{Maker-Maker}, \gamename{Maker-Breaker}, and \gamename{Enforcer-Avoider}.
One of our main tools is a new fragment of first-order logic where universally-quantified variables only occur in inequalities and no other relations.
After giving some necessary definitions in the next section, we will state our results precisely, and discuss them. The rest of the paper is devoted to the proofs of our results, with some parts deferred to the appendix, due to space constraints.

\section{Preliminaries}
\subparagraph*{Finite structures.}
A vocabulary $\tau$ is a finite set of relation symbols, each having an associated arity.
A finite structure $\mathcal{A}$ over $\tau$ consists of a finite set $A$, called the universe, and for each $R$ in $\tau$ a relation over $A$ of corresponding arity.
An (undirected) graph is a finite structure $G=(V,E)$, where $E$ is a symmetric binary relation.
A hypergraph is a finite structure $G=(V\cup E,\mathit{IN})$, where $\mathit{IN} \subseteq V \times E$ is the incidence relation between vertices and edges.
Sometimes it is more convenient to denote a hypergraph instead by a tuple $G=(V,E)$ where $E$ is a set of subsets of $V$.

\subparagraph*{First-order logic.}
We assume a countably infinite set of variables.
Atomic formulas over vocabulary $\tau$ are of the form $x_1=x_2$ or $R(x_1,\dots,x_k)$ where $R\in \tau$ and $x_1,\dots,x_k$ are variables.
The class \FO\ of all first-order formulas over $\tau$ consists of formulas that are constructed from atomic formulas over $\tau$ using standard Boolean connectives $\neg, \wedge, \vee$ as well as quantifiers $\exists,\forall$ followed by a variable.
Let $\varphi$ be a first-order formula.
The size of (a reasonable encoding of) $\varphi$ is denoted by $\card{\varphi}$.
The variables of $\varphi$ that are not in the scope of a quantifier are called free variables.
We denote by $\varphi(\mathcal{A})$ the set of all assignments of elements of $A$ to the free variables of $\varphi$ such that $\varphi$ is satisfied.
We call $\mathcal{A}$ a model of $\varphi$ if $\varphi(\mathcal{A})$ is not empty.
The class $\Sigma_1$ contains all first-order formulas of the form $\exists x_1, \dots, \exists x_k \varphi$ where $\varphi$ is a quantifier free first-order formula.

\subparagraph*{Parameterized complexity.}
The class \fpt contains all parameterized problems that can be decided by an \fpt-algorithm.
An \fpt-algorithm is an algorithm with running time $f(k)\cdot n^{\mathcal{O}(1)}$, where $f(\cdot)$ is an arbitrary computable function that only depends on the parameter $k$ and $n$ is the size of the problem instance.
An \textit{\fpt-reduction} of a parameterized problem $\Pi$ to a parameterized problem $\Pi'$ is an \fpt-algorithm that transforms an instance $(I, k)$ of $\Pi$ to an instance $(I', k')$ of $\Pi'$ such that: (i) $(I, k)$ is a yes-instance of $\Pi$ if and only if $(I', k')$ is a yes-instance of $\Pi'$, and (ii) $k' = g(k)$, where $g(\cdot)$ is an arbitrary computable function that only depends on $k$.
Hardness and completeness with respect to parameterized complexity classes is defined analogously to the concepts from classical complexity theory, using \fpt-reductions.
The following parameterized classes will be needed in this paper:
$ \fpt \subseteq \wone \subseteq \awstar$.
Many parameterized complexity classes can be defined via a version of the following model checking problem.
\paramdecisionproblem{
  \textsc{MC($\Phi$)}}{
  Finite structure $\mathcal{A}$ and formula $\varphi \in \Phi$.}{
  $\card{\varphi}$.}{
  Decide whether $\varphi(\mathcal{A}) \neq \emptyset$.}
In particular, the problem \textsc{MC($\Sigma_1$)} is \wone-complete and the problem \textsc{MC(\FO)} is \awstar-complete (see for example~\cite{FlumG1998}).

\subparagraph*{Positional games.}
Positional games are played by two players on a hypergraph $G=(V,E)$.
The vertex set $V$ indicates the set of available positions, while the each hyperedge $e\in E$ denotes a winning configuration.
For some games, the hyperedges are implicitly defined, instead of being explicitly part of the input.
The two players alternatively claim unclaimed vertices of $V$ until either all elements are claimed or one player wins.
A \emph{position} in a positional game is an allocation of vertices to the players, who have already claimed these vertices.
The \emph{empty position} is the position where no vertex is allocated to a player.
The notion of winning depends on the game type.
In a \emph{Maker-Maker game}, the first player to claim all vertices of some hyperedge $e \in E$ wins.
In a \emph{Maker-Breaker game}, the first player (\emph{Maker}) wins if she claims all vertices of some hyperedge $e\in E$.
If the game ends and player 1 has not won, then the second player (\emph{Breaker}) wins.
In an \emph{Enforcer-Avoider game}, the first player (\emph{Enforcer}) wins if the second player (\emph{Avoider}) claims all vertices of some hyperedge $e\in E$.
If the game ends and player 1 has not won, then the second player wins.
A positional game is called an $l$-move game, if the game ends either after a player wins or both players played $l$ moves.
A winning strategy for player 1 is a move for player 1 such that for all moves of player 2 there exists a move of player 1\dots such that player 1 wins.

\section{Results}

The first game we consider is a Maker-Maker game that generalizes well-known games \textsc{Tic-Tac-Toe}, \textsc{Connect6}, and \textsc{Gomoku} (also known as \textsc{Five in a Row}).
In \Connect, the vertices are cells of an $m \times n$ grid, each set of $k$ aligned cells (horizontally, vertically, or diagonally) is a winning set, the first move by player 1 is to claim $q$ vertices, and then the players alternate claim $p$ unclaimed vertices at each turn. \textsc{Tic-Tac-Toe} corresponds to \ConnectNoP($3,3,3,1,1$), \textsc{Connect6} to \ConnectNoP($19,19,6,2,1$), and \textsc{Gomoku} to \ConnectNoP($19,19,5,1,1$). Variations with different board sizes are also common.
In the \SConnect problem, the input is the set of $m\cdot n$ vertices, an assignment of some of these vertices to the two players, the integer $p$, and the parameter $\ell$.
The winning sets corresponding to the $k$ aligned cells are implicitly defined.
The question is whether player 1 has a winning strategy from this position in at most $\ell$ moves.
We omit $q$ from the problem definition of \SConnect since we are modeling games that advanced already past the initial moves.
Our first result (proved in Section \ref{sec:connect}) is that \SConnect is fixed-parameter tractable for parameter $\ell$. (In all our results, the parameter is the number of moves, $\ell$.)

\begin{restatable}{theorem}{thmconnect}
  \label{thm:connect}
  \SConnect is \fpt.
\end{restatable}

The main reason for this tractability is the rather special structure of the winning sets. It helps reducing the problem to model checking for first-order logic on locally bounded treewidth structures, which is \fpt \cite{FrickG2001}.

A similar strategy was recently used to show that \SHex is \fpt~\citep{BonnetJS2016TCS}.
The \Hex game is played on a parallelogram board paved by hexagons, each player owns two opposite sides of the parallelogram. Players alternately claim an unclaimed cell, and the first player to connect their sides with a path of connected hexagons wins the game.
Note that we may view \Hex as a Maker-Breaker game: if the second player manages to disconnect the first players sides, he has created a path connecting his sides.
\citet{BonnetJS2016TCS} also considered a well-known generalization to arbitrary graphs.
The $\GHex$ game is played on a graph with two specified vertices $s$ and $t$. The two players alternately claim an unclaimed vertex of the graph, and player 1 wins if she can connect $s$ and $t$ by vertices claimed by her, and player 2 wins if he can prevent player 1 from doing so.
The \SGHex problem has as input a graph $G$, two vertices $s$ and $t$ in $G$, an allocation of some of the vertices to the players, and an integer $\ell$. The parameter is $\ell$, and the question is whether player 1 has a winning strategy to connect $s$ and $t$ in $\ell$ moves.

The \SGHex problem is known to be in \awstar and was conjectured to be \awstar-complete \cite{BonnetJS2016TCS,DowneyF2013,DowneyF99,FominM12,Scott2009}.
In fact, \awstar is thought of as the natural home for most short games \cite{DowneyF2013}, playing a similar role in parameterized complexity as PSPACE in classical complexity for games with polynomial length.
However, \citet{BonnetJS2016TCS} only managed to show that \SGHex is \wone-hard, leaving a complexity gap between \wone and \awstar.
Our next result is to show that \SGHex is in \wone.
Thus, \SGHex is in fact \wone-complete.
\begin{restatable}{theorem}{thmhex}
  \label{thm:hex}
  \SGHex is \wone-complete.
\end{restatable}
Our main tool is a new fragment of first-order logic for which model-checking on arbitrary relational structures is \wone-complete parameterized by the length of the formula.
This fragment, which we call \neqFO, is the fragment of first-order logic where universally-quantified variables appear only in inequalities.
\begin{restatable}{theorem}{thmcustomized}
  \label{thm:customized}
  \Customized is \wone-complete.
\end{restatable}
This result is proved by reducing a formula in \neqFO to a formula in \Sigmaone.
The \neqFO logic makes it convenient to express short games where we can express that player 1 can reach a certain configuration without being blocked by player 2, no matter what configurations player 2 reaches.
This is indeed the case for \GHex, where we are merely interested in knowing if player 1 can connect $s$ and $t$ without being blocked by player 2.

More generally, this is the case for \SMB, where the input is a hypergraph $G=(V,E)$, a position, and an integer $\ell$, and the question is whether player 1 has a winning strategy to claim all the vertices of some hyperedge in $\ell$ moves.
\begin{restatable}{theorem}{thmmb}
  \label{thm:mb}
  \SMB is \wone-complete.
\end{restatable}
The fact that \SMB is \pspace-complete and \wone-complete (and \emph{not} \awstar-complete) may challenge the intuition one has on alternation.
Looking at the classical complexity (\pspace-completeness), it seems that both players have comparable expressivity and impact over the game.
As the game length is polynomially bounded, if the outcome could be determined by only guessing a sequence of moves from one player, then the problem would lie in \np.
Now from the parameterized complexity standpoint, \SMB is equivalent under FPT reductions to guessing the $k$ vertices of a clique (as in the seminal \wone-complete \textsc{$k$-Clique} problem); no alternation there.
Those considerations may explain why it was difficult to believe that \GHex is \emph{not} \awstar-complete as conjectured repeatedly \cite{Scott2009,DowneyF99,DowneyF2013}.

This is also in contrast to \SMM, where the input is a hypergraph $G=(V,E)$, a position, and an integer $\ell$, and the question is whether player 1 has a strategy to be the first player claiming all the vertices of some hyperedge in $\ell$ moves.
\begin{restatable}{theorem}{thmmm}
  \label{thm:mm}
  \SMM is \awstar-complete.
\end{restatable}

For the remaining type of positional games, the \SEA problem has as input a hypergraph $G=(V,E)$, a position, and an integer $\ell$, and the question is whether player 1 has a strategy to claim $\ell$ vertices that forces player 2 to complete a hyperedge. Again, player 1 can only block some moves of player 2, and the winning condition for player 1 can be expressed in \neqFO.
\begin{restatable}{theorem}{thmea}
  \label{thm:ea}
  \SEA is \cowone-complete.
\end{restatable}
Our results suggest that a structured board may suggest that a positional game is \fpt, but otherwise, the complexity depends on how the winning condition for player 1 can be expressed. If it only depends on what positions player 1 has reached, our results suggest that the problem is \wone-complete, but when the winning condition for player 1 also depends on the position player 2 has reached, the game is probably \awstar-complete.

\section{\SConnect is \fpt}
\label{sec:connect}
Graph $G$ represents an $m \times n$ board in the following sense.
Every board cell is represented by a vertex.
Horizontal, vertical and diagonal neighbouring cells are connected via an edge.
Vertex sets $V_1$ and $V_2$ represent the vertices already occupied by Player 1 and Player 2.
While integer $p$, the number of stones to be placed during a move, is part of the input, we restrict it to values below constant $k$ as games with $p\geq k$ are trivial.

\paramdecisionproblem{
  \SConnect}{
  A graph $G=(V,E)$ representing an $m \times n$ board, occupied vertices $V_1, V_2 \subseteq V$, and integer $p$ and $l$.}{
  $l$.}{
  Decide whether Player 1 has a winning strategy with at most $l$ moves.}

\thmconnect*
\begin{proof}
  We reduce \SConnect to first-order model checking \textsc{MC(\FO)} on a bounded local treewidth structure.
  Using a result by Frick and Grohe~\cite{FrickG2001}, it follows that \SConnect is \fpt.
  Let $(G, V_1, V_2, p, l)$ be an instance of \SConnect, where $G = (V,E)$.
  We construct instance $(\mathcal{A},\varphi)$ of \textsc{MC(\FO)} as follows.
  Let $\mathit{EDGE}$ be a binary relation symbol and let $\mathit{V1}$ and $\mathit{V2}$ be unary relation symbols.
  Then $\mathcal{A}$ is the $\{\mathit{EDGE},\mathit{V1},\mathit{V2}\}$-structure $(V,\mathit{EDGE}^\mathcal{A},\mathit{V1}^\mathcal{A},\mathit{V2}^\mathcal{A})$ with $\mathit{EDGE}^\mathcal{A} \coloneqq E$, $\mathit{V1}^\mathcal{A} \coloneqq V_1$, and $\mathit{V2}^\mathcal{A} \coloneqq V_2$.
  FO-formula $\varphi$ is defined as $\varphi \equiv \exists x_1^1 \exists x_1^2 \dots \exists x_1^p \forall y_1^1 \dots \forall y_1^p \exists x_2^1 \dots \exists x_2^p \forall y_2^1 \dots \exists x_l^p
  \exists u_1 \exists u_2 \dots \exists u_k \forall v_1 \forall v_2 \dots \forall v_k \psi$,  
  \begin{align*}
  \psi \equiv \bigvee_{i = 0}^l \Big[ & \mathit{legalP1}_i(x_1^1,\dots,x_1^p,y_1^1,\dots,x_l^p) \land \Big( \neg \mathit{legalP2}_i(x_1^1,\dots,x_1^p,y_1^1,\dots,x_l^p) \vee \\
    & \quad \Big( \mathit{configP1}_i(x_1^1,\dots,x_l^p,u_1,\dots,u_k) \land \bigwedge_{j = 1}^{k - 2} \gAligned{u_{j},u_{j + 1},u_{j + 2}} \land \\
    & \quad\quad \Big( \neg \mathit{configP2}_i(y_1^1,\dots,y_l^p,v_1,\dots,v_k) \lor \neg \bigwedge_{j = 1}^{k - 2} \gAligned{v_{j},v_{j + 1},v_{j + 2}} \Big) \Big) \Big) \Big]
  \end{align*}
  \vspace*{-0.5cm}
  \begin{align*}
    \gpath{u,v,w} \equiv& \, \mathit{EDGE}(u,v) \land \mathit{EDGE}(v,w),\\
    \galigned{u,v,w} \equiv& \, \exists x \exists y \, \gpath{u,v,w} \land \gpath{u,x,w} \land \gpath{u,y,w} \land \gpath{x,v,y} \land\\
        & \forall z \Big[ \big( z \neq v \land z \neq x \land z \neq y \big) \rightarrow \neg \gpath{u,z,w} \Big],\\
    \gdaligned{u,v,w} \equiv& \, \gpath{u,v,w} \land \forall x \Big[ x \neq v \implies \neg \gpath{u,x,w} \Big],\\
    \gAligned{u,v,w} \equiv& \, \galigned{u,v,w} \lor \gdaligned{u,v,w}.
  \end{align*}
  Variables $x_i^j$ represent the $j$th stone in Player 1's $i$th move and variables $y_i^j$ represent the $j$th stone in Player 2's $i$th move.
  The sequences $u_1 \dots u_k$ and $v_1 \dots v_k$ represent possible winning configurations for Player 1 and Player 2.
  The overall structure of $\psi$ is the following.
  The first disjunction ranging from $i=0$ to $i=l$ represents the number of moves Player 1 needs to win the game.
  We then ensure that the $x$ variables represent legal moves by Player 1.
  Further, either variables $y$ do not represent legal moves by Player 2, or Player 1 achieved a winning configuration.
  For the latter, we assure that variables $u$ represent aligned vertices occupied by Player 1.
  Finally, we check that Player 2 did not achieve a winning configuration before, that is vertices $v$ do not represent aligned vertices occupied by Player 2.
  
  Formula $\gpath{u,v,w}$ expresses that there is a path of length 2 between vertices $u$ and $w$ via $v$ ($\mathit{configP1}_i$ and $\mathit{configP2}_i$ ensure that the arguments are disjoint vertices).
  Formula $\galigned{u,v,w}$ expresses that vertices $u$, $v$, and $w$ are aligned horizontally or vertically in this order.
  A case analysis shows that $u, v$ and $w$ are horizontally or vertically aligned if and only if there are exactly three nodes at distance 1 of $u$ and $w$, and that $v$ is in the middle of the other two.
  In case $u, v$ and $w$ are located on one of the border lines of the board, there are exactly two nodes at distance 1.
  Formula $\gdaligned{u,v,w}$ expresses that vertices $u$, $v$, and $w$ are diagonally aligned in this order.
  This is the case if there exists no other length 2 path between $u$ and $w$.
  Formula $\gAligned{u,v,w}$ expresses that vertices $u$, $v$, and $w$ are aligned (in that order).
  Formula $\mathit{legalP1}_i$ (see Appendix~\ref{app:subformulas-connect}) ensures that variables $x_i^j$ represent legal moves of Player 1, that is vertices not contained in $V_1$ or $V_2$ or previously played vertices.
  Analogously, $\mathit{legalP2}_i$ ensures that variables $y_i^j$ represent legal moves of Player 2.
  Formula, $\mathit{configP1}_i$ (see Appendix~\ref{app:subformulas-connect}) expresses that variables $u_1,\dots,u_k$ form a valid configuration of exactly $k$ vertices out of the set of $V_1$ or vertices played by Player 1.
  Analogously, $\mathit{configP2}_i$ states that variables $v_1,\dots,v_k$ form a valid configuration of exactly $k$ vertices out of the set of $V_2$ or vertices played by Player 2.
  The size of $\varphi$ is polynomial in $l$, $k$, and $p$.
  Since $k$ is a constant and $p$ is bounded by $k$, we have an FO formula polynomial in our parameter $l$.
  Graph $G$ represents a grid with diagonals.
  Hence, $G$ has maximum degree 8.
  It follows from Seese~\cite{Seese1996} that Short Connect is \fpt.
\end{proof}

\section{\Customized is \wone-complete}
The class \neqFO contains all first-order formulas of the form $Q_1 x_1 Q_2 x_2 Q_3 x_3 \dots Q_k x_k \varphi$, with $Q_i \in \{ \forall , \exists \}$ and $\varphi$ being a quantifier free first-order formula such that every $\forall$-quantified variable $x_i$ only occurs in inequalities, that is in relations of the form $x_i \neq x_j$ for some variable $x_j$.
Furthermore, $\varphi$ does not contain any other variables besides $x_1,\dots,x_k$.

\thmcustomized*
\begin{proof}
  Hardness: Every $\Sigma_1$ formula is contained in the class \neqFO.
  Hence, \wone-hardness follows from \wone-completeness of \Mosigmaone.

  \begin{sloppypar}
  Membership:
  By reduction to \Mosigmaone.
  Let $(\mathcal{A},\varphi)$ be an instance of \Customized.
  If $\varphi$ contains only existential quantifiers then $(\mathcal{A},\varphi)$ is already an instance of \Mosigmaone.
  Hence, let $\varphi = Q_1 x_1 Q_2 x_2 \dots Q_{i-1} x_{i-1} \forall x_i \exists x_{i+1} \exists x_{i+2} \dots \exists x_k \psi$ with $Q_j \in \{\forall,\exists\}$ for $1 \leq j < i$, $\psi$ is in negation normal form and $\card{\varphi}=l$.
  That is, $x_i$ is the rightmost of the universal quantified variables.
  In order to reduce $(\mathcal{A},\varphi)$ to an instance of \Mosigmaone, we need a way to remove all universal quantifications.
  We will show how to eliminate the universal quantification of $x_i$.
  This technique can then be used to iteratively eliminate all the universal quantifiers.
  Let $\varphi_1(x_1,\dots,x_{i-1})$ be the subformula $\varphi_1(x_1,\dots,x_{i-1}) = \forall x_{i} \exists x_{i+1} \dots \exists x_k \psi$.
  We will show that we can replace $\varphi_1(x_1,\dots,x_{i-1})$ by
  \begin{align}
    \varphi_2(x_1,\dots,x_{i-1}) =&\, \exists y_{i} \exists y_{i+1} \dots \exists y_{k} \Big( \psi[y_{i}/x_{i},y_{i+1}/x_{i+1}, \dots ,y_k/x_k] \land\\
    & \bigwedge_{j = 1}^{i-1} \exists y_{i+1}^j \exists y_{i+2}^j \dots \exists y_k^j  \psi[x_j/x_{i}, y_{i+1}^j/x_{i+1}, y_{i+2}^j/x_{i+2}, \dots , y_k^j/x_k] \land\\
     & \bigwedge_{j = i + 1}^k \exists y_{i+1}^j \exists y_{i+2}^j \dots \exists y_k^j  \psi[y_j/x_{i}, y_{i+1}^j/x_{i+1}, y_{i+2}^j/x_{i+2}, \dots , y_k^j/x_k] \Big).
  \end{align}
  This reduction is an \fpt-reduction, since the size of formula $\varphi_2$ is a function of the size of formula $\varphi_1$.
  Let $c_1,\dots,c_{i-1}$ be arbitrary but fixed elements of the universe $A$ of $\mathcal{A}$.
  We will show that $\varphi_1(x_1,\dots,x_{i-1}) \equiv \varphi_2(x_1,\dots,x_{i-1})$ by proving (a) $\varphi_1(c_1,\dots,c_{i-1}) \rightarrow \varphi_2(c_1,\dots,c_{i-1})$ and (b) $\varphi_2(c_1,\dots,c_{i-1}) \rightarrow \varphi_1(c_1,\dots,c_{i-1})$.
  For (a) assume that $\varphi_1(c_1,\dots,c_{i-1})$ is true.
  This means, $\varphi_1[c_i/x_i]$ is true for all $c_i\in A$, that is for all $c_i\in A$ there exists an assignment to $x_{i+1},\dots,x_k$ such that $\psi$ is true.
  Part (1) of $\varphi_2(c_1,\dots,c_{i-1})$ asks for some $c_i \in A$ such that there exists an assignment to $x_{i+1},\dots,x_k$ such that $\psi$ is true.
  Part (2) asks for the existence of an assignment to $x_{i+1},\dots,x_k$ such that $\psi$ is true for each of the cases where $x_i$ is one of the elements $c_1,\dots,c_{i-1}$.
  Part (3) asks for the existence of an assignment to $x_{i+1},\dots,x_k$ such that $\psi$ is true for each of the cases where $x_i$ is one of the elements that are assigned to $x_{i+1},\dots,x_k$ in the model of Part (1).
  All these are special cases of the universal quantification over $x_i$, hence $\varphi_2(c_1,\dots,c_{i-1})$ is true.
  \end{sloppypar}

  For direction (b) assume towards a contradiction that $\varphi_1(c_1,\dots,c_{i-1})$ is false and that $\varphi_2(c_1,\dots,c_{i-1})$ is true.
  Since $\varphi_1$ is false, there exists $c_i \in A$ such that $\varphi_1[c_i/x_i]$ is false.
  We perform a case distinction on the value $c_i$.
  First let $c_i = c_j$ for some $j \in \{1,\dots,i-1\}$.
  Then let $c_{i+1},\dots,c_k$ be the assignments to variables $y_{i+1}^j,\dots,y_k^j$ in the model of $\varphi_2$.
  The $j$th conjunct of Part (2) of $\varphi_2$ states that $\psi$ holds for $x_i=x_j$ using the assignment $c_{i+1},\dots,c_k$.
  Hence, assigning $c_{i+1},\dots,c_k$ to variables $x_{i+1},\dots,x_k$ in $\varphi_1$ is a model for $\varphi_1[c_i/x_i]$, which contradicts our assumption.
  As the next case, let $c_{i+1},\dots,c_k$ be the assignment to variables $y_{i+1},\dots,y_k$ in the model of $\varphi_2$ and let $c_i = c_j$ for some $j\in \{i+1,\dots,k\}$.
  Let $c_{i+1}',\dots,c_k'$ be the assignments to variables $y_{i+1}^j,\dots,y_k^j$ in the model of $\varphi_2$.
  The conjunct with index $j$ of Part (3) of $\varphi_2$ states that $\psi$ holds for $x_i=x_j=c_j$ using the assignment $c_{i+1}',\dots,c_k'$.
  Hence, assigning $c_{i+1}',\dots,c_k'$ to variables $x_{i+1},\dots,x_k$ in $\varphi_1$ is a model for $\varphi_1[c_i/x_i]$, which contradicts our assumption.
  For the last case, let $c_i$ be one of the remaining values.
  Let $l_1,\dots,l_m$ be all the literals in $\psi$ that contain $x_i$.
  All of them are inequalities of the form $x_i \neq x_j$ for $j \neq i$.
  Let $c_i'$ be the assignment to $y_i$ in the model of $\varphi_2$.
  Let $l_1',\dots,l_m'$ be the literals in $\psi[y_{i}/x_{i},y_{i+1}/x_{i+1}, \dots ,y_k/x_k]$ in Part (1) of $\varphi_2$ that correspond to $l_1,\dots,l_m$.
  We have no knowledge about the truth value of these literals $l_j'$ with $1 \leq j \leq m$, but all of the literals $l_j$ in $\psi$ evaluate to true when assigning $c_{i+1},\dots,c_k$ to the variables $x_{i+1},\dots,x_k$.
  Since $\psi$ is in negation normal form and the literals $l_1,\dots,l_m$ never occur in unnegated form, that is as equalities, changing the truth value of these literal from false to true will never result in changing the truth value of the whole formula from true to false.
  But since $c_i'$ together with $c_{i+1},\dots,c_k$ is a model of Part (1) of $\varphi_2$, it holds that for all values of $c_i$ that we consider in this case, that $\varphi_1[c_i/x_i]$ is true, which contradicts our assumption.
  This completes the case distinction and we have $\varphi_1(x_1,\dots,x_{i-1}) \equiv \varphi_2(x_1,\dots,x_2)$.
\end{proof}

\section{\SGHex is \wone-complete}
\paramdecisionproblem{
  \SGHex}{
  Graph $G=(V,E)$, vertices $s,t\in V$, vertex sets $V_1,V_2\subseteq V$ with $V_1\cap V_2 = \emptyset$, and integer $l$.}{
  $l$.}{
  Decide whether Player 1 has a winning strategy with at most $l$ moves in the generalized Hex game $(G,s,t,V_1,V_2)$.}

A generalized Hex game $(G,s,t,V_1,V_2)$ is a positional game $(V',E')$, where the positions $V'$ and the winning configurations $E'$ are defined as follows.
Set $V'$ contains all vertices of $G$, that is $V' = V$.
Set $E'$ contains a set of vertices $\{v_1,\dots,v_k\}$ if and only if $\{v_1,\dots,v_k\} \cup \{s,t\}$ form an $s-t$ path in $G$.
Additionally, vertices in $V_1$ and $V_2$ are already claimed by player 1 and player 2, respectively.
Since the set of winning configurations of \SGHex is only defined implicitly, the input size of \SGHex can be exponential smaller than the number of winning configurations.

\thmhex*
\begin{proof}
  Hardness is already known~\cite{BonnetJS2016TCS}.
  For membership, we reduce \SGHex to \Customized.
  Let $(G,s,t,V_1,V_2,l)$ be an instance of \SGHex, where $G = (V,E)$.
  Claimed vertices $V_1$ and $V_2$ can be preprocessed: (i) every $v\in V_1$ and its incident edges are removed from $G$ and the neighbourhood of $v$ is turned into a clique;
  (ii) every $v\in V_2$ and its incident edges are removed from $G$.
  Hence, w.l.o.g.\ we assume that $V_1=V_2=\emptyset$.
  We construct an instance $(\mathcal{A},\varphi)$ of \Customized as follows.
  Let $\mathit{EDGE}$ be a binary relation symbol and let $S$ and $T$ be unary relation symbols.
  Then $\mathcal{A}$ is the $\{\mathit{EDGE},S,T\}$-structure $(V,\mathit{EDGE}^\mathcal{A},S^\mathcal{A},T^\mathcal{A})$ with $\mathit{EDGE}^\mathcal{A} \coloneqq E$, $S^\mathcal{A} \coloneqq \{s\}$, and $T^\mathcal{A} \coloneqq \{t\}$.
  The \neqFO-formula $\varphi$ is defined as $\varphi = \exists s \exists t \exists x_1 \forall y_1 \exists x_2 \forall y_2 \dots \forall y_{l-1} \exists x_l \exists z_1 \exists z_2 \dots \exists z_l \psi$, with
  \vspace*{-7mm}
  \begin{flalign*}
    \psi \equiv S(s) \wedge T(t) \wedge \Big( \mathit{EDGE}(s, t) \lor & \bigvee_{i=1}^l \bigvee_{j=1}^i \Big( \mathit{EDGE}(s,z_1) \land \mathit{EDGE}(z_j,t) \land&\\
    & \mathit{path}_{i,j}(x_1,\dots,x_i,z_1,\dots,z_j) \land  \mathit{diff}_i(x_1,y_1,\dots,y_{i-1},x_i) \Big) \Big),&
  \end{flalign*}
  \vspace*{-0.9cm}
  \begin{flalign*}
    \mathit{path}_{i,j}(x_1,\dots,x_i,z_1,\dots,z_j) \equiv& \bigwedge_{h=1}^{j-1} \mathit{EDGE}(z_h,z_{h+1}) \land \bigwedge_{h=1}^j \bigvee_{k=1}^i z_h = x_k,&\\
    \mathit{diff}_i(x_1,y_1,\dots,x_{i-1},y_{i-1},x_i) \equiv& \bigwedge_{1 \leq j < k \leq i} x_j \neq x_k \land \bigwedge_{1 \leq j < k \leq i} y_j \neq x_k.&
  \end{flalign*}
  \vspace*{-0.5cm}

  The intuition of $\varphi$ is the following.
  The variables $x_i$, $y_i$, and $z_i$ represent the moves of Player 1, the moves of Player 2, and the ordered $(s,t)$-path induced by Player 1's moves, respectively.
  The variables $s$ and $t$ represent the vertices of the same name.
  Formula $\varphi$ expresses that there is either a direct edge between $s$ and $t$ or a $s$-$t$ path of length $j$ was played.
  The main disjunctions ($\bigvee$) ensure that we consider wins that take up to $l$ moves, and build $s$-$t$ path of length up to $l$.
  Subformula $\mathit{path}_{i,j}$ will be true if and only if the $z$ variables form a path using only values of the selected values for the $x$ variables.
  Subformula $\mathit{diff}_i$ ensures that all $x$ variables are pairwise distinct and they are distinct from all $y$ variables with smaller index.

  We have $\card{\varphi} = \mathcal{O}(l^4)$, so this is indeed an \fpt-reduction and \wone-membership follows.
\end{proof}

\section{\SMB is \wone-complete}

\paramdecisionproblem{
  \SMB}{
  Hypergraph $G=(V,E)$, vertex sets $V_1,V_2\subseteq V$ with $V_1\cap V_2 = \emptyset$, and integer $l$.}{
  $l$.}{
  Decide whether Player 1 has a winning strategy with at most $l$ if vertices $V_1$ and $V_2$ are already claimed by Player 1 and Player 2, respectively.}

\thmmb*
\begin{proof}
  For membership, we reduce \SMB to \Customized.
  Let $(G,V_1,V_2,l)$ be an instance of \SMB, where $G = (V,E)$ is a hypergraph.
  Claimed vertices $V_1$ and $V_2$ can be preprocessed: (i) every $v\in V_1$ is removed from $V$ and every hyperedge $e\in E$;
  (ii) every $v\in V_2$ is removed from $V$ and every hyperedge $e\in E$ with $v \in e$ is removed from $E$.
  Hence, w.l.o.g.\ we assume that $V_1=V_2=\emptyset$.
  We construct an instance $(\mathcal{A},\varphi)$ of \Customized as follows.
  Let $\mathit{IN}$ and $\mathit{SIZE}$ be binary relation symbols.
  Then $\mathcal{A}$ is the $\{\mathit{IN},\mathit{SIZE}\}$-structure $(V\cup E \cup \{1,\dots,\card{V}\},\mathit{IN}^\mathcal{A},\mathit{SIZE}^\mathcal{A})$ with $\mathit{IN}^\mathcal{A} \coloneqq \{(x,e) \mid x \in V, e \in E, x \in e\}$ and $\mathit{SIZE}^\mathcal{A} \coloneqq \{(e,i) \mid e \in E, \card{e}=i \}$.
  Hence, the universe of $\mathcal{A}$ consists of the vertices of $G$, an element for each hyperedge, and an element for some bounded number of integers.
  The \neqFO-formula $\varphi$ is defined as $\varphi \equiv \exists x_1 \forall y_1 \dots \forall y_{l-1} \exists x_l \exists e \exists z_1 \exists z_2 \dots \exists z_l \psi$, with
\begin{flalign*}
  & \psi \equiv \hspace{-1mm}\bigvee_{1 \leq j \leq i \leq l}\hspace{-1mm} \Big(\mathit{diff}_i(x_1,y_1,\dots,x_i) \land \mathit{SIZE}(e,j) \land \bigwedge_{k = 1}^j \bigvee_{h = 1}^i z_k = x_h \land \hspace{-2mm}\bigwedge_{1 \leq k < h \leq j}\hspace{-2mm} z_k \neq z_h \land \bigwedge_{k = 1}^j \mathit{IN}(z_k,e)
  \Big). &
\end{flalign*}
The subformula $\mathit{diff}_i(x_1,y_1,\dots,x_i)$ refers to the subformula with same name used in the proof of Theorem~\ref{thm:hex}.
That is, it ensures that all $x$ variables are pairwise distinct and that they are distinct from all $y$ variables with smaller index.
The intuition of $\varphi$ is the following.
The variables $x_i$ and $y_i$ represent the moves of Maker and the moves of Breaker, respectively.
The variables $z_i$ represent the vertices forming the winning configuration of Maker and $e$ represents the hyperedge of this winning configuration.
The first disjunction ensures that we consider wins that take up to $l$ moves.
The second disjunction ensures that we consider winning configurations that consist of up to $i$ vertices.
After checking that $e$ has the correct size ($\mathit{SIZE}(e,j)$), we encode that the values of the $z$ variables are contained in the hyperedge represented by $e$ and that these variables are pairwise disjoint and selected among the moves of Maker (the $x$ variables).

We have $\card{\varphi} = \mathcal{O}(l^4)$, so this is indeed an \fpt-reduction and \wone-membership follows.

  For hardness, we reduce \CClique to \SMB.
  The reduction is essentially the same as the reduction used for showing \wone-hardness of \GHex~\cite{BonnetJS2016TCS}.
  The crucial observation is that the construction of \citet{BonnetJS2016TCS} contains only a polynomial number of possible $s-t$ paths.
  Hence, we can encode every such $s-t$-path as a unique hyperedge denoting a winning configuration in \SMB.
\end{proof}

\section{\SMM is \awstar-complete}
\paramdecisionproblem{
  \SMM}{
  Hypergraph $G=(V,E)$, vertex sets $V_1,V_2\subseteq V$ with $V_1\cap V_2 = \emptyset$, and integer $l$.}{
  $l$.}{
  Decide whether Player 1 has a winning strategy with at most $l$ if vertices $V_1$ and $V_2$ are already claimed by Player 1 and Player 2.}

\thmmm*
\begin{proof}
  For membership, we reduce \SMM to \MCFO.
  Let $(G,V_1,V_2,l)$ be an instance of \SMM, where $G = (V,E)$ is a hypergraph.
  We construct an instance $(\mathcal{A},\varphi)$ of \MCFO as follows.
  Let $\mathit{V1}$, $\mathit{V2}$, and $\mathit{EDGE}$ be unary relation symbols.
  Let $\mathit{IN}$ be a binary relation symbol.
  Then $\mathcal{A}$ is the $\{\mathit{V1}, \mathit{V2}, \mathit{EDGE}, \mathit{IN}\}$-structure $(V\cup E, \mathit{V1}^\mathcal{A}, \mathit{V2}^\mathcal{A}, \mathit{EDGE}^\mathcal{A}, \mathit{IN}^\mathcal{A})$ with $\mathit{V1}^\mathcal{A} \coloneqq V_1$, $\mathit{V2}^\mathcal{A} \coloneqq V_2$, $\mathit{EDGE}^\mathcal{A} \coloneqq E$, and $\mathit{IN}^\mathcal{A} \coloneqq \{(x,e) \mid x \in V, e \in E, x \in e\}$.
  Hence, the universe of $\mathcal{A}$ consists of the vertices and the hyperedges of $G$.
  The \FO-formula $\varphi$ is defined as $\varphi \equiv \exists x_1 \forall y_1 \dots \forall y_{l-1} \exists x_l \psi$, with
\vspace{-4mm}
\begin{flalign*}
  \psi \equiv \bigvee_{i = 0}^l \mathit{legalP1}_i(x_1, y_1, \dots, x_l) \land \Big( & \neg\mathit{legalP2}_{i-1}(x_1, y_1, \dots, x_l) \vee &\\
    & \big(\mathit{winP1}_i(x_1,y_1,\dots, x_l) \wedge \neg\mathit{winP2}_{i-1}(x_1,y_1,\dots, x_l) \big)\Big).&
\end{flalign*}
\\[-15mm]
\begin{flalign*}
  &\mathit{winP1}_i(x_1,y_1,\dots,x_l) \equiv \exists e \forall z \mathit{EDGE}(e) \wedge \big(\neg \mathit{IN}(z, e) \vee \mathit{V1}(z) \vee \bigvee_{j = 1}^i z = x_j\big),&\\
  &\mathit{winP2}_i(x_1,y_1,\dots,x_l) \equiv \exists e \forall z \mathit{EDGE}(e) \wedge \Big(\neg \mathit{IN}(z, e) \vee \mathit{V2}(z) \vee \bigvee_{j = 1}^i z = y_j\Big).
\end{flalign*}
Variable $x_j$ represent Player 1's $j$th move and variable $y_j$ represent Player 2's $j$th move.
The first disjunction represents the number of moves $i$ that Player 1 needs to win the game.
Formula $\mathit{legalP1}_i$ (see Appendix~\ref{app:subformulas-mm}) ensures that variables $(x_j)_{1 \leq j \leq i}$ represent legal moves of Player 1, that is vertices not contained in $V_1$ or $V_2$ or previously played vertices.
Analogously, $\mathit{legalP2}_i$ ensures that variables $(y_j)_{1 \leq j \leq i}$ represent legal moves of Player 2.
Formula $\mathit{winP1}_i$ ensures that Player 1 has won within the $i$ first moves, that is, it has completed a hyperedge with $V_1$ and variables up to $x_i$.
Analogously, $\mathit{winP2}_i$ ensures that Player 2 has won within the $i$ first moves.
We have $\card{\varphi} = \mathcal{O}(l^3)$ and $\card{\mathcal{A}}= \mathcal{O}(\card{G}^2)$, so this is indeed an \fpt-reduction and \awstar-membership follows.

For hardness, we reduce from the \awstar-complete problem \SGG on bipartite graphs.
The reduction is deferred to the appendix.
\end{proof}

\section{\SEA is \cowone-complete}
\paramdecisionproblem{
  \SEA}{
  Hypergraph $G=(V,E)$, vertex sets $V_1,V_2\subseteq V$ with $V_1\cap V_2 = \emptyset$, and integer $l$.}{
  $l$.}{
  Decide whether Player 1 has a winning strategy with at most $l$ moves if vertices $V_1$ and $V_2$ are already claimed by Player 1 and Player 2, respectively.}

\thmea*
\begin{proof}
  We show that the co-problem of \SEA is \wone-complete.
  The co-problem of \SEA is to decide whether for all strategies of Enforcer, there exists a strategy of Avoider such that during the first $l$ moves, Avoider does not claim a hyperedge.
  Again, vertices $V_1$ and $V_2$ are already claimed by Enforcer and Avoider, respectively.
  We prove \wone-hardness by a parameterized reduction from \textsc{Independent Set} and \wone-membership by reduction to \Customized.

  In the \wone-complete \textsc{Independent Set} problem \cite{DowneyF99}, the input is a graph $G=(V,E)$ and an integer parameter $k$, and the question is whether $G$ has an independent set of size $k$, i.e., a set of $k$ pairwise non-adjacent vertices.
  We construct a positional game $G'=(V',E')$ by replacing each vertex of $G$ by a clique of size $k+1$.
  The vertex set $V'$ has vertices $v(1),\dots,v(k+1)$ for each vertex $v\in V$, and hyperedges are $E' = \{\{v(i),v(j)\} : v\in V \text{ and } 1\le i<j\le k+1\} \cup \{\{u(i),v(j)\} : uv\in E \text{ and } 1\le i,j\le k+1\}$.
  We claim that $G$ has an independent set of size $k$ if and only if Avoider does not claim a hyperedge in the first $k$ moves in the positional game $G'$ starting from the empty position, that is $V_1 = V_2 = \emptyset$.
  For the forward direction, suppose $I=\{v_1,\dots,v_k\}$ is an independent set of $G$ of size $k$.
  Then, a winning strategy for Avoider is to claim an unclaimed vertex from $\{v_i(1),\dots,v_i(k+1)\}$ at round $i\in\{1,\dots,k\}$.
  We note that Enforcer cannot claim all the vertices from $\{v_i(1),\dots,v_i(k+1)\}$, since there are not enough moves to do so, and Avoider does not complete a hyperedge with this strategy.
  On the other hand, suppose Avoider has a winning strategy in $k$ moves.
  For an arbitrary play by Enforcer, let $\{v_1(i_1),\dots,v_k(i_k)\}$ denote the vertices claimed by Player 1.
  Then, $v_i \ne v_j$ and $v_iv_j\notin E$ for any $1\le i<j\le k$, since Player 1 would otherwise claim all the vertices of a hyperedge.
  Therefore, $\{v_1,\dots,v_k\}$ is an independent set of $G$ of size $k$.

  For membership, we reduce to \Customized.
  Let $(G,V_1,V_2,l)$ be an instance of the co-problem of \SEA where $G = (V,E)$ is a hypergraph.
  First we do some preprocessing.
  We remove all vertices from $G$ that are contained in $V_2$, that is the vertices already claimed by Avoider.
  If this results in an empty hyperedge, the instance is a no-instance.
  Otherwise, we remove all hyperedges that contain a vertex in $V_1$, that is the vertices already claimed by Enforcer, since Avoider will never lose via these edges anymore.
  Finally, we remove all vertices from $G$ that are contained in $V_1$.
  Let $G = (V,E)$ now refer to the outcome of this preprocessing.
  By construction all vertices of $G$ are unoccupied and some vertices might not be contained in any hyperedge.
  If $G$ contains less than $2 l$ vertices we can solve the problem via brute force in \fpt time.
  Hence, in what follows we assume that there are at least $2 l$ unoccupied vertices in $G$.
  We construct an instance $(\mathcal{A},\varphi)$ of \Customized as follows.
  Let $\mathit{EDGE_i}$ be a $i$-ary relation symbol for $1 \leq i \leq l$.
  Then $\mathcal{A}$ is the $\{\mathit{EDGE_1},\dots,\mathit{EDGE_l}\}$-structure $(V,\mathit{EDGE_1}^\mathcal{A},\dots,\mathit{EDGE_l}^\mathcal{A})$ with $\mathit{EDGE_i}^\mathcal{A} \coloneqq \{(v_1,\dots,v_i) \mid e \in E, \card{e}=i, e = \{v_1,\dots,v_l\}\}$, that is $\mathit{EDGE_i}^\mathcal{A}$ contains every permutation of all hyperedges of cardinality $i$.
  The \neqFO-formula $\varphi$ is defined as
\begin{equation*}
  \varphi \equiv \forall y_1 \exists x_1 \forall y_2 \exists x_2 \dots \exists x_l \; \mathit{diff}_l(y_1,x_1,\dots,x_l) \land \bigwedge_{1 \leq i \leq l} \bigwedge_{\{z_1,\dots,z_i\}\subseteq\{x_1,\dots,x_l\}} \neg \mathit{EDGE_i}(z_1,\dots,z_i),
\end{equation*}
where $\mathit{diff}_i(y_1,x_1,\dots,x_i) \equiv \bigwedge_{1 \leq j < k \leq i} x_j \neq x_k \land \bigwedge_{1 \leq j \leq k \leq i} y_j \neq x_k.$

  Subformula $\mathit{diff}_i(y_1,x_1,\dots,x_i)$ ensures that all $x$ variables are pairwise distinct and they are distinct from all $y$ variables with index less or equal theirs.
  The intuition of $\varphi$ is the following.
  The variables $x_i$ and $y_i$ represent the moves of Avoider and the moves of Enforcer, respectively.
  Avoider wins if the $x$ variables do not cover a whole hyperedge after $l$ moves.
  We only have to check hyperedges of size up to $l$.
  Hence, for each cardinality $i \leq l$, we check for all subsets $z_1,\dots, z_l$ of the $x$ variables that they do not form a hyperedge.
  Formula $\varphi$ does not pose any restrictions on the $y$ variables, that is we do not force Enforcer to pick unoccupied vertices.
  We call a move by Enforcer that picks an already occupied vertex cheating.
  To prove correctness, we need to show that whenever Enforcer has a winning strategy $\sigma_E$ that involves cheating, Enforcer also has a winning strategy $\sigma'_E$ without cheating.
  We construct $\sigma'_E$ as follows.
  We follow strategy $\sigma_E$ while $\sigma_E$ does not perform a cheating move.
  If the next move would be a cheating move, we play a random unoccupied vertex instead and keep track of this vertex in a new set $V_r$.
  The next time we need to select a move, we construct a board state $s$ by removing all vertices in $V_r$ from the picks of Enforcer and query strategy $\sigma_E$ on this state $s$.
  If the answer is an unoccupied vertex, we perform this move normally.
  If instead the answer is a previously played vertex (which might be in $V_r$), we play a random unoccupied vertex instead and add it to $V_r$.
  Since $\sigma_E$ was a winning strategy, so is $\sigma'_E$.
  Hence, formula $\varphi$ does not need to check if the $y$ variables correspond to unoccupied vertices.
  The construction can be done by an \fpt algorithm since for each hyperedge $e \in E$ of cardinality $i$, we create $i!\leq l!$ entries in the $\mathit{EDGE_i}$ relation.
  We have $\card{\varphi} = \mathcal{O}(l^l)$, so this is indeed an \fpt reduction and \wone-membership follows.
\end{proof}

\section{Conclusion}
We have seen that the parameterized complexity of short positional games depends crucially on whether both players compete for achieving winning sets, or whether the game can be seen as one player aiming to achieve a winning set and the other player merely blocking the moves of the first player. Naturally, blocking moves correspond to inequalities in first-order logic, and our \neqFO fragment of first-order logic therefore captures that the universal player can only block moves of the existential player. Our \wone-completeness of \Customized has been used several times in this paper, but our transformation of \neqFO formulas into \Sigmaone formulas may have other uses.
As a concrete example related to positional games, \citet{BonnetJS2016TCS} established that \SHex is \fpt by expressing the problem as a \FO\ formula, and making use of Frick and Grohe's meta-theorem \cite{FrickG2001}, similarly as we did in Section \ref{sec:connect}.
This establishes that the problem is \fpt but the running time is non-elementary in $l$.
However, we remark that their \FO\ formula is actually a \neqFO formula of size polynomial in $l$. Our transformation gives an equivalent \Sigmaone formula whose length is single-exponential in $l$, and the meta-theorem of \citet{GroheW2004} then gives a running time for solving \SHex that is  triply-exponential in $l$.

\paragraph*{Acknowledgments}
We thank anonymous reviewers for helpful comments and we thank Yijia Chen and Paul Hunter for bringing \citeauthor{GroheW2004}'s work to our attention.
Serge Gaspers is the recipient of an Australian Research Council (ARC) Future Fellowship (FT140100048).
Abdallah Saffidine is the recipient of an ARC DECRA Fellowship (DE150101351).
This work received support under the ARC's Discovery Projects funding scheme (DP150101134).

\bibliographystyle{plainurl}%

\appendix
\section{Subformulas for Theorem~\ref{thm:connect}}
\label{app:subformulas-connect}

  \begin{align*}
    \mathit{legalP1}_i(x_1^1,\dots,x_1^p,y_1^1,\dots,x_l^p) \equiv \bigwedge_{j = 1}^i \bigwedge_{t = 1}^p \Big[ & \neg \mathit{V1}(x_j^t) \land \neg \mathit{V2}(x_j^t) \land \bigwedge_{r = 1}^{j-1} \bigwedge_{q = 1}^{t} (x_j^t \neq x_r^q) \land \\
    & \bigwedge_{q = 1}^{t - 1} (x_j^t \neq x_j^q) \land \bigwedge_{r = 1}^{j-1} \bigwedge_{q = 1}^{t} (x_j^t \neq y_r^q) \Big],\\
    \mathit{legalP2}_i(x_1^1,\dots,x_1^p,y_1^1,\dots,x_l^p) \equiv \bigwedge_{j = 1}^{i-1} \bigwedge_{t = 1}^p \Big[ & \neg \mathit{V1}(y_j^t) \land \neg \mathit{V2}(y_j^t) \land \bigwedge_{r = 1}^{j-1} \bigwedge_{q = 1}^{t} (y_j^t \neq y_r^q) \land \\
    & \bigwedge_{q = 1}^{t - 1} (y_j^t \neq y_j^q) \land \bigwedge_{r = 1}^{j-1} \bigwedge_{q = 1}^{t} (y_j^t \neq x_r^q) \Big].
  \end{align*}

  \begin{align*}
     \mathit{configP1}_i(x_1^1,\dots,x_l^p,u_1,\dots,u_k) \equiv & \bigwedge_{j = 1}^k \left[ \left( \mathit{V1}(u_j) \lor \bigvee_{r = 1}^i \bigvee_{q = 1}^p u_j = x_r^q \right) \land \bigwedge_{r = 1}^{j-1} u_j \neq u_r \right],\\
     \mathit{configP2}_i(y_1^1,\dots,y_l^p,v_1,\dots,v_k) \equiv & \bigwedge_{j = 1}^k \left[ \left( \mathit{V2}(v_j) \lor \bigvee_{r = 1}^{i-1} \bigvee_{q = 1}^p v_j = y_r^q \right) \land \bigwedge_{r = 1}^{j-1} v_j \neq v_r \right].\\
  \end{align*}

\section{Subformulas for Theorem~\ref{thm:mm}}
\label{app:subformulas-mm}

\begin{flalign*}
  \mathit{legalP1}_i(x_1,y_1,\dots,x_l) \equiv& \bigwedge_{1 \leq j \leq i} \Big[ \neg \mathit{V1}(x_j) \land \neg \mathit{V2}(x_j) \Big] \wedge \bigwedge_{1 \leq j < k \leq i} \Big[ x_j \neq x_k \land y_j \neq x_k \Big],&\\
  \mathit{legalP2}_i(x_1,y_1,\dots,x_l) \equiv& \bigwedge_{1 \leq j \leq i} \Big[ \neg \mathit{V1}(y_j) \land \neg \mathit{V2}(y_j) \Big] \wedge \hspace{-1mm}\bigwedge_{1 \leq j < k \leq i} (y_j \neq y_k) \land \hspace{-1mm}\bigwedge_{1 \leq j \leq k \leq i} \hspace{-1mm}x_j \neq y_k.\\
\end{flalign*}
\section{\awstar-hardness of \SMM}

Reduction from the \awstar-complete problem \SGG on bipartite graphs.
\SGG is played by two players on a bipartite graph.
Players alternate in picking a vertex that is a neighbour of the previously picked vertex of the opponent.
A vertex can only be picked, if it has not already been picked during the game.
A player loses if there is no legal move left for her.

\paramdecisionproblem{
  \SGG}{
  Bipartite graph $(X \uplus Y,F)$, start vertex $v_0 \in X$ and integer $k$.}{
  $k$.}{
  Decide whether Player 1 has a winning strategy that needs at most $k$ moves.}

From an instance $B=(X \uplus Y, F, v_0), k$ of \SGG, with $v_0 \in X$, we build a hypergraph $G=(V, E), l$ of size polynomial in $|B|$ which will be an equivalent \SMM instance.

In our reduction, the hypergraph $G$ mainly involves two distinguished vertices $\exists, \forall \in V$ and gadgets corresponding to vertices and edges of $B$.
In the initial setup, the vertex $\exists$ is assumed to have already been claimed by Player 1 and the vertex $\forall$ to have already been claimed by Player 2.
Our construction ensures that all the hyperedges of $E$ contain exactly one vertex in $\{\exists, \forall\}$.
We thus partition the hyperedges between the ones that can make Player 1 win and the ones that can make Player 2 win.

Formally, $G$ is defined as indicated in Equations~\eqref{eq:MMdefV} and~\eqref{eq:MMdefE}.
It uses gadgets $\vexists{\cdot}$, $\vforall{\cdot}$, $V^D_{4}(\cdot)$, $E^{\exists}(\cdot)$, $E^\forall(\cdot)$, $D_{4}^{\exists}(\cdot)$, $D_{4}^{\forall}(\cdot)$ detailed in the rest of this section.
The parameter is linearly preserved from the input parameter: $l=9(k+1)+6$.
\begin{align}
\label{eq:MMdefV}
V & = \{\exists, \forall\} \cup \bigcup_{u \in X} \vexists{u} \cup \bigcup_{u \in Y} \vforall{u} \cup \vdelay{4}{\exists} \cup \vdelay{4}{\forall}\\
\label{eq:MMdefE}
E & = \{\{\forall, a^{v_0}\}\} \cup \bigcup_{u \in X} E^{\exists}(u) \cup \bigcup_{u \in Y} E^\forall(u) \cup  \delae{4}{\exists} \cup \delaa{4}{\forall}
\end{align}

$V_1 = \{\exists\}, V_2 = \{\forall\}$

\subsection{Terminology}
A \emph{useless} 3-threat for Player 1 is a 3-threat that can be defended, and for which after the 3-threat and its defense, Player 1 has not achieved anything.
Formally, the threat and its defense are two vertices which, once played, do not appear in any other hyperedges that could make one player or their opponent win.
Note that those threats can be disregarded for Player 1 but not for Player 2.
Indeed, Player 2 could use a series of useless 3-threats to win by delaying the game.

A \emph{losing} 3-threat for a player is a 3-threat that can be met with a counter-attack winning in a constant number of moves; more precisely in at most 6 moves.

A \emph{living} 3-threat is a non losing 3-threat; if it is for Player 1, it should in addition be non useless.

\subsection{Delay gadget}
As a building block of the forthcoming existential and universal gadgets, we introduce the following delay gadgets where $? \in \{\exists, \forall\}$.
If $?=\exists$ (resp. $?=\forall$), we say that the delay gadget belongs to Player 1 (resp. to Player 2).
\begin{flalign}
&\delai{?}{1}{S} := \{S \cup \{?, x^S_1\}, S \cup \{?, x^S_2\}, S \cup \{?, x^S_3\}\} &\\
&\delai{?}{2}{S} := \hspace{-1mm}\bigcup_{i,j \in [3]}\hspace{-1mm} \{S \cup \{?,x^S_i,y^S_j\}\} =
 \left\{\begin{array}{@{}c@{}}
 S \cup \{?, x^S_1, y^S_1\}, S \cup \{?, x^S_2, y^S_1\}, S \cup \{?, x^S_3, y^S_1\},\\
 S \cup \{?, x^S_1, y^S_2\}, S \cup \{?, x^S_2, y^S_2\}, S \cup \{?, x^S_3, y^S_2\},\\
 S \cup \{?, x^S_1, y^S_3\}, S \cup \{?, x^S_2, y^S_3\}, S \cup \{?, x^S_3, y^S_3\}\end{array}\right\} &\\
&\delai{?}{4}{S} := \bigcup_{g,h,i,j \in [3]} \{S \cup \{?,x^S_g,y^S_h,z^S_i,t^S_j\}\}&
\end{flalign}
\begin{flalign}
&\vdelayb{1}{S} := \{x^S_1, x^S_2, x^S_3\} &\\
&\vdelayb{2}{S} := \bigcup_{i \in [3]} \{x^S_i,y^S_i\} = \{x^S_1, x^S_2, x^S_3, y^S_1, y^S_2, y^S_3\}&\\
&\vdelayb{4}{S} := \bigcup_{i \in [3]} \{x^S_i,y^S_i,z^S_i,t^S_i\}&
\end{flalign}
The elements $x^S_i$, $y^S_i$, $z^S_i$, and $t^S_i$ (with $i \in [3]$) will only appear in the corresponding delay gadgets.
For any set $S$, we will introduce at most one set among $\delaab{1}{S}$, $\delaeb{1}{S}$, $\delaab{2}{S}$, $\delaeb{2}{S}$, $\delaab{4}{S}$, and $\delaeb{4}{S}$.
This implies that existing $x^S_i$ and $y^S_i$ (with $i \in [3]$) are well-defined.

\begin{lemma}
  \label{lem:delay1}
  Let $\delta\in\{1,2,4\}$ and $S \subseteq V$ be a set of vertices such that $\delaeb{\delta}{S} \subseteq E$ (resp. $\delaab{\delta}{S} \subseteq E$).
  If all vertices of $S$ have been claimed by Player 1 (resp. Player 2), and if no more than one vertex of $\vdelayb{\delta}{S}$ has been claimed by the opponent, then she (resp. he) has an unstoppable $\delta$-threat.
\end{lemma}
\begin{proof}
The two statements have identical proofs by switching Player 1 and Player 2.
We therefore only give a proof for a delay gadget $\delaeb{\delta}{S}$.
Assume that Player 1 has played all the vertices of $S$.
Without loss of generality, assume that the vertex claimed by the opponent, if any, is $x^S_1$.
Recall that we assume that $\exists$ has already been claimed by Player 1 and $\forall$ has been claimed by Player 2.

For $\delta = 1$,
Player 1 has at least two 1-threats, playing in $x^S_2$ or $x^S_3$, and Player 2 cannot block them both.
Thus, if Player 2 claims $x^S_i$ (with $i \in [3]$), she claims $x^S_j$ with $j \neq i \in [3]$ and wins.

For $\delta = 2$,
Player 1 has several 2-threats.
If Player 2 claims $x^S_i$ (resp. $y^S_i$) for some $i \in [3]$, Player 1 claims $x^S_j$ (resp. $y^S_j$) for some $j \neq i \in \{2,3\}$ and obtains an unstoppable $1$-threat.

For $\delta = 4$, the reasoning is similar and omitted.
\end{proof}

\begin{corollary}
  \label{cor:delay}
  Let $\delta\in\{1,2,4\}$ and $S \subseteq V$ be a set of vertices such that $\delaeb{\delta}{S} \subseteq E$ (resp. $\delaab{\delta}{S} \subseteq E$).
  If Player 1 (resp.~Player 2) claims all vertices in $S$ and no more than one vertex of $\vdelayb{\delta}{S}$ has been claimed by the opponent, then if it is that player's turn, they can force a win in $\delta$ moves unless the opponent has a $\delta-1$-threat.
  If it is the opponent's turn, then Player 1 (resp.~Player 2) can force a win in $\delta$ moves unless the opponent has a $\delta$-threat.
\end{corollary}

\subsection{Existential vertex gadget}

For each vertex $u \in X$ in the existential partition of the \SGG instance, we introduce in $G$ the following hyperedges:
\begin{align*}
E^{\exists}(u) = & \quad \delae{2}{a^u, b^u} \cup \delaa{2}{b^u, e^u} \cup \delaa{2}{b^u, g^u} \\
& \cup \bigcup_{v\in N(u)} \delae{1}{a^u, c^u_v, d^u_v} \\
 & \cup \delaa{1}{b^u, d^u_v, e^u} \cup \delae{1}{c^u_v, e^u, f^u}\\
 & \cup \delaa{1}{d^u_v, f^u, g^u} \cup \delae{1}{c^u_v, g^u, h^u}\\
 & \cup \delaa{2}{d^u_v, i^u_v} \cup \delae{2}{i^u_v, a^v}
\end{align*}

In terms of vertices of $G$ introduced by the gadget, each vertex $u \in X$ gives rise to a set $\vexists{u}$ that contains all the vertices needed by the delay sub-gadgets along with $ \{a^u, b^u, e^u, f^u, g^u, h^u\} \cup \bigcup_{v \in N(u)} \{c^u_v, d^u_v, i^u_v\}$.

\begin{lemma}
\label{lem:existential-exists}
Consider the gadget for an existential vertex $u \in X$ such that no element of $\vexists{u}\setminus \{a^u\}$ has been claimed yet.
Assume that $a^u$ has been played by Player 1 and that it is Player 2's turn.
If Player 2 has no non-losing 3-threats in the whole board, then for each $v \in N(u)$ such that $a^v$ has not been claimed yet, Player 1 has a strategy $\sigma^{\exists}(u,v)$ that ensures either that Player 2 plays $a^v$ after no more than 8 moves all of which belonging to $\vexists{u}$ and that there are no non-losing 3-threats left for Player 2 in the gadget or that Player 1 wins in no more than 14 moves.
\end{lemma}
\begin{proof}
We exhibit the strategy for Player 1 and show that Player 2's answers are forced to prevent Player 1 from winning.
By assumption Player 2 has no non-losing 3-threats anywhere else on the board and no vertices claimed in $\vexists{u}$, so unless Player 2 play $b^u$, Player 1 wins in 6 moves by claiming $b^u$ herself via Corollary~\ref{cor:delay} applied to $\delae{2}{a^u, b^u}$.
Although Player 2 has now sets of 3-threats which involve $e^u$ and $g^u$, he does not have any 2-threats.
Player 1 plays $c^u_v$ which forces Player 2 to claim $d^u_v$ by Corollary~\ref{cor:delay} applied to $\delae{1}{a^u, c^u_v, d^u_v}$.
Player 1 plays $e^u$ which forces Player 2 to claim $f^u$.
Player 1 plays $g^u$ which forces Player 2 to claim $h^u$.
Player 1 plays $i^u_v$.
At this point, 8 moves have been played, Player 2 has no 3-threats left in the gadget, so Player 2 is forced to play $a^v$ lest Player 1 plays $a^v$ and wins in a total of 14 moves by Corollary~\ref{cor:delay} applied to $\delae{2}{i^u_v, a^v}$.

Since Player 1 has claimed $e^u$, $g^u$, and $i^u_v$, the only local hyperedges remaining for Player 2 are $\delaa{2}{d^u_w, i^u_w}$ for $w\neq v$, and none of them feature a 3-threat.
\end{proof}

\begin{lemma}
\label{lem:existential-forall}
Consider the gadget for an existential vertex $u \in X$ such that no element of $\vexists{u}\setminus \{a^u\}$ has been claimed yet.
Assume that for any vertex $v \in Y$, $a^v$ has not been claimed by Player 1.
Assume that $a^u$ has been played by Player 1 and that it is Player 2's turn.
If Player 1 has no living 3-threats elsewhere on the board, then Player 2 has a strategy $\sigma^{\forall}(u)$ that ensures either 1) that after no less than 8 moves, all of which either belong to $\vexists{u}$ or are not in any live existential hyperedge, Player 2 plays $a^v$ for some $v$ and there are no living 3-threats left for Player 1, or it is Player 2's turn and there is no living 3-threat for Player 1; or 2) that Player 2 wins.
\end{lemma}
\begin{proof}
We exhibit a local strategy for Player 2, any move by Player 1 in a non-living 3-threat elsewhere on the board is responded to accordingly.
Player 2 plays $b^u$ creating sets of 3-threats in $\delaa{2}{b^u, e^u}$ and $\delaa{2}{b^u, g^u}$.
Playing either of $e^u$ and $g^u$ is losing for Player 1 because Player 2 can play in the other vertex.
Therefore, Player 1 needs to play in a 3-threat to avoid losing.
Notwithstanding the non-living 3-threats, the only 3-threats for Player 1 can be found in the gadgets $\delae{1}{a^u, c^u_v, d^u_v}$ for $v \in N(u)$.

As long as Player 1 plays in $d^u_w$ for some $w$, Player 2 replies in the corresponding $c^u_w$ voiding the threat.
As soon as Player 1 plays a move other than $d^u_v$ in $\delae{1}{a^u, c^u_v, d^u_v}$ for some $v$, Player 2 can answer $d^u_v$, voiding the threat, and play proceeds as follows.
Player 1 has no 2-threats and so replying $e^u$ is forced to avoid losing via Corollary~\ref{cor:delay} applied to $\delaa{1}{b^u,d^u_v,e^u}$.
Player 2 plays $f^u$ which forces Player 1 to claim $g^u$.
Player 2 plays $h^u$ threatening to play $i^u_v$.
Therefore, Player 1 needs to either play in 3-threats via the gadgets $\delae{1}{a^u, c^u_w, d^u_w}$ for some $w \in N(u)$ such that $d^u_w$ has not been claimed yet, or Player 1 has to play in $i^u_v$ herself.
As long as Player 1 plays in $d^u_w$ for some $w$, Player 2 replies in the corresponding $c^u_w$ voiding the threat.

Eventually, Player 1 has to play in $i^u_v$.
If $a^v$ has already been claimed by Player 2, then Player 2 is left with no 3-threat to defend.
Otherwise, Player 2 plays $a^v$.
\end{proof}

\subsection{Universal vertex gadget}
For each $u \in Y$, we introduce in $G$ the following hyperedges:

\begin{align*}
E^{\forall}(u) = & \quad \delaa{2}{a^u,b^u} \cup \delae{2}{b^u, g^u} \cup \delae{2}{b^u, i^u}\\
  & \cup \bigcup_{v\in N(u)} \delaa{1}{a^u, c^u_v, d^u_v}\\
  & \cup \delae{1}{b^u, d^u_v, e^u_v} \cup \delaa{1}{c^u_v, e^u_v, f^u}\\
  & \cup \delae{1}{b^u, c^u_v, j^u_v} \cup \delaa{1}{d^u_v, j^u_v, f^u}\\
  & \cup \delae{1}{c^u_v, f^u, g^u} \cup \delae{1}{d^u_v, f^u, g^u}\cup \delaa{1}{e^u_v, g^u, h^u}\\
  & \cup \delae{1}{f^u, h^u, i^u} \cup \delaa{2}{e^u_v, a^v}
\end{align*}

In terms of vertices of $G$ introduced by the gadget, each vertex $u \in Y$ gives rise to a set $\vforall{u}$ that contains all the vertices needed by the delay sub-gadgets along with $\{a^u, b^u, f^u, g^u, h^u, i^u\} \cup \bigcup_{v \in N(u)} \{c^u_v, d^u_v, e^u_v, j^u_v\}$.

We observe that the only shared vertices between the different existential and universal gadgets are $a^u$ for $u \in X \cup Y$.
For instance, in the universal gadget, each $a^v$ with $v \in N(u)$ is the ``starting vertex'' of the existential gadget encoding the vertex $v \in X$.

\begin{lemma}
\label{lem:universal-exists}
Consider the gadget for a universal vertex $u \in Y$ such that no element of $\vforall{u}\setminus \{a^u\}$ has been claimed yet.
Assume that for any vertex $v \in X$, $a^v$ has not been claimed by Player 2.
Assume that $a^u$ has been played by Player 2 and that it is Player 1's turn.
If Player 2 has no non-losing 3-threats elsewhere on the board, then Player 1 has a strategy $\sigma^{\exists}(u)$ that ensures either 1) that after no more than 8 moves, all of which belong to $\vforall{u}$, Player 1 plays $a^v$ for some $v$ and there are no non-losing 3-threats left for Player 2 or it is Player 1's turn and there are no non-losing 3-threats for Player 2; or 2) that Player 1 wins in no more than 14 moves.
\end{lemma}
\begin{proof}
We exhibit a local strategy for Player 1.
Player 1 plays $b^u$ creating sets of 3-threats in $\delae{2}{b^u, g^u}$ and $\delae{2}{b^u, i^u}$.
Claiming either of $g^u$ and $i^u$ is losing for Player 2 because Player 1 can play in the other vertex.
Therefore, Player 2 needs to play in a 3-threat to avoid losing.
The only non-losing 3-threats for Player 2 can be found in the gadget $\delaa{1}{a^u, c^u_v, d^u_v}$ for $v \in N(u)$.

If Player 2 claims $d^u_v$, Player 1 plays $c^u_v$, forcing Player 2 to claim $j^u_v$.
Player 1 plays $f^u$, forcing Player 2 to claim $g^u$.
At this point, Player 1 can play $i^u$ and win by Corollary~\ref{cor:delay} applied to $\delae{2}{b^u,i^u}$.

If instead of $d^u_v$ Player 2 starts by claiming $c^u_v$, then Player 1 plays $d^u_v$, forcing Player 2 to claim $e^u_v$.
Player 1 plays $f^u$, forcing Player 2 to claim $g^u$.
Player 1 plays $h^u$, forcing Player 2 to claim $i^u$.
If $a^v$ has already been claimed by Player 1, then Player 1 is left with no 3-threat to defend.
Otherwise, Player 1 plays $a^v$.
\end{proof}

\begin{lemma}
\label{lem:universal-forall}
Consider the gadget for a universal vertex $u \in Y$ such that no element of $\vforall{u}\setminus \{a^u\}$ has been claimed yet.
Assume that $a^u$ has been played by Player 2 and that it is Player 1's turn.

If Player 1 has no living 3-threats on the whole board, then for each $v \in N(u)$ such that $a^v$ has not been claimed yet, Player 2 has a strategy $\sigma^{\forall}(u,v)$ that ensures either that Player 1 plays $a^v$ after no less than 8 moves all of which either belong to $\vforall{u}$ or are not in any live existential hyperedge and that there are no living 3-threats left for Player 1 in the gadget; or that Player 2 wins.
\end{lemma}
\begin{proof}
We exhibit the strategy for Player 2 and show that Player 1's answers are forced to prevent Player 2 from winning.
By assumption Player 1 has no living 3-threats anywhere else on the board and no vertices claimed in $\vforall{u}$, so unless Player 1 plays $b^u$, Player 2 wins in 6 moves by claiming $b^u$ himself via Corollary~\ref{cor:delay} applied to $\delaa{2}{a^u, b^u}$.
Although Player 1 has now sets of 3-threats which involve $e^u$ and $g^u$, she does not have any 2-threats.
Player 2 plays $c^u_v$ which forces Player 1 to claim $d^u_v$ by Corollary~\ref{cor:delay} applied to $\delaa{1}{a^u, c^u_v, d^u_v}$.
Player 2 plays $e^u$ which forces Player 1 to claim $f^u$.
Player 2 plays $g^u$ which forces Player 1 to claim $h^u$.
Player 2 plays $i^u$.
At this point, 8 moves have been played, Player 1 has no 3-threats left in the gadget, so Player 1 is forced to play $a^v$ lest Player 2 plays $a^v$ and wins in a total of 14 moves by Corollary~\ref{cor:delay} applied to $\delaa{2}{e^u_v, a^v}$.

Since Player 2 has claimed $g^u$ and $i^u$, the only local hyperedges remaining for Player 1 are in $\delae{2}{b^u, d^u_w, e^u_w}$ and $\delae{2}{b^u, c^u_w, j^u_w}$ for $w\neq v$, and none of them is feature a living 3-threat.
\end{proof}

\subsection{Correctness of the reduction}\label{sec:correctness}

To show that YES \SGG instances are mapped to YES \SMM instances and that NO instances are mapped onto NO instances, we prove that any Player 1 winning strategy in \SGG gives rise to a winning strategy for Player 1 in the corresponding \SMM instance, and conversely for Player 2 winning/delaying strategies.

Assume that Player 1 can ensure a win within $k$ moves in \SGG with strategy $\tau$, and let us construct a strategy $\sigma$ ensuring a Player 1 win within $l$ moves in \SMM.
After Player 1 starts with move $a^{v_0}$, we use $\tau$, Lemma~\ref{lem:existential-exists}, and Lemma~\ref{lem:universal-exists} to create $\sigma$ such that whenever $\tau$ prescribes that the token moves from a vertex $u \in X$ to $v \in Y$, we use $\sigma^{\exists}(u,v)$ to leave the $u$-gadget and enter the $v$-gadget.
When the \SMM game enters a $u$-gadget with $u \in Y$, we use $\sigma^{\exists}(u)$ to select moves in the gadget until the $u$-gadget is left and enters a $v$-gadget with $v \in X$.
If $a^v$ is already claimed by Player 1, then Player 2 has no non-losing threats and Player 1 can enter the $\delae{4}{\exists}$ gadget and win by Corollary~\ref{cor:delay}.
Otherwise, we then update the \SGG game with Player 2 moving the token to $v$.
Eventually, the \SGG game reaches a vertex $u \in Y$ such that all neighbors have been visited before and the game ends.
In the \SMM instance, Player 1 follows $\sigma^{\exists}(u)$ and then wins by entering the $\delae{4}{\exists}$ gadget.
If $\tau$ guarantees that at most $k' \leq k$ moves are played before Player 1 wins, then $\sigma$ guarantees that at most $9(k'+1)+6 \leq l$ moves are played before Player 1 wins.

In the case of a NO \SGG instance, Player 2 has a strategy $\tau$ such that either Player 2 wins, or the game goes for longer than $k$ moves.
A corresponding \SMM strategy $\sigma$ can be derived such that either Player 2 wins in the \SMM game, or the game goes for longer than $9(k+1)+6 = l$ moves.
The construction is dual to the one above and relies on Lemma~\ref{lem:existential-forall} and Lemma~\ref{lem:universal-forall}.

\end{document}